\documentclass{article}
\usepackage[utf8]{inputenc}
\usepackage{stmaryrd}

\usepackage{hyperref}

\usepackage{vmargin}

\usepackage[english]{babel}

\usepackage{amsmath, amssymb, mathrsfs}

\usepackage{amsthm}

\usepackage[all]{xy}
\usepackage{stmaryrd}
\usepackage[usenames,dvipsnames]{xcolor}

\usepackage{subfigure}
\usepackage{color}
\definecolor{gris}{gray}{0.45}

\usepackage{xspace}   

\usepackage{microtype}

\usepackage{multirow}

\usepackage{tikz}
\usetikzlibrary{calc, matrix, backgrounds,shapes.geometric}
\usepackage{dsfont}

\usepackage[ruled, vlined]{algorithm2e}
\usepackage{algorithmic}

\newfont{\mycrnotice}{ptmr8t at 7pt}
\newfont{\myconfname}{ptmri8t at 7pt}
\begin{document}

\newtheorem{theo}{Theorem}[section]
\newtheorem{lem}[theo]{Lemma}
\newtheorem{prop}[theo]{Proposition}
\newtheorem{cor}[theo]{Corollary}
\newtheorem{conj}{Conjecture}
\newtheorem{quest}[theo]{Question}

\theoremstyle{definition}
\newtheorem{deftn}[theo]{Definition}
\newtheorem{ex}[theo]{Example}

\theoremstyle{remark}
\newtheorem{rmk}[theo]{Remark}
\newtheorem{note}[theo]{Note}

\title{Matrix-F5 Algorithms and Tropical \\
Gröbner Bases Computation}

\author{Tristan Vaccon \\ Université de Rennes 1 \\ tristan.vaccon@univ-rennes1.fr}
\date{}

\maketitle
\begin{abstract}
Let $K$ be a field equipped with a valuation. Tropical varieties over $K$ can be defined with a theory of Gröbner bases taking into account the valuation of $K$. Because of the use of the valuation, this theory is promising for stable computations over polynomial rings over a $p$-adic fields.

We design a strategy to compute such tropical Gröbner bases by adapting the Matrix-F5 algorithm. Two variants of the Matrix-F5 algorithm, depending on how the Macaulay matrices are built, are available to tropical computation with respective modifications. The former is more numerically stable while the latter is faster.

Our study is performed both over any exact field with valuation and some inexact fields like $\mathbb{Q}_p$ or $\mathbb{F}_q \llbracket t \rrbracket.$ In the latter case, we track the loss in precision, and show that the numerical stability can compare very favorably to the case of classical Gröbner bases when the valuation is non-trivial. Numerical examples are provided.
\end{abstract}




\section{Introduction}

Despite its young age, tropical geometry has revealed to be of significant value, with applications in algebraic geometry, combinatorics, computer science, and non-archimedean geometry (see \cite{Maclagan}, \cite{EINSIEDLER}).

Effective computation over tropical varieties make decisive usage of Gröbner bases, but before Chan and Maclagan's definition of tropical Gröbner bases taking into account the valuation in \cite{Chan}, \cite{Chan-Maclagan}, computations were only available over fields with trivial valuation where standard Gröbner bases techniques applied.
In this document, we show that following this definition, Matrix-F5 algorithms can be performed to compute tropical Gröbner bases.

Our motivations are twofold. Firstly, our result bears promising application for computation over fields with valuation that are not effective, such as $\mathbb{Q}_p$ or $\mathbb{Q} \llbracket t \rrbracket.$ Indeed, in \cite{Vaccon}, the author studies computation of Gröbner bases over such fields and proves that for a regular sequence and under some regularity assumption (whose genericity is at best conjectural) and with enough initial entry precision, approximate Gröbner bases can be computed. Thank to the study of Matrix-F5 algorithms, we prove that to compute a tropical Gröbner basis of the ideal generated by $F=(f_1,\dots,f_s)$, $F$ being regular and known with enough initial precision is sufficient. Hence, generically, approximate tropical Gröbner bases can be computed. Moreover, for a special choice of term order, the smallest loss in precision that can be obtained by linear algebra is attained: tropical Gröbner bases then provide a generically numerically stable alternative to Gröbner bases.

Secondly, Matrix-F5 algorithms allow an easy study of the complexity of the computation of tropical Gröbner bases and are a first step toward a tropical F5 algorithm.

\textbf{Related works on tropical Gröbner bases:} We refer to the book of Maclagan and Sturmfels \cite{Maclagan} for an introduction to computational tropical algebraic geometry.

The computation of tropical varieties over $\mathbb{Q}$ with trivial valuation is available in the Gfan package by Anders Jensen (see \cite{Jensen}), by using standard Gröbner basis computation.
Yet, for computation of tropical varieties over general fields, with non-trivial valuation, such techniques are not readily available. 
This is why Chan and Maclagan have developed in \cite{Chan-Maclagan} a way to extend the theory of Gröbner bases to take into account the valuation and allow tropical computation. Their theory of tropical Gröbner bases is effective and allows, with a suitable division algorithm, a Buchberger algorithm.

\textbf{Main results:} Let $K$ be a field equipped with a valuation $val$. Let $\geq$ be an order on the terms of $K[X_1,\dots, X_n]$ as in Definition \ref{trop order}, defined with $w \in Im(val)^n$ and a monomial ordering $\geq_1.$
Following \cite{Chan-Maclagan}, we define tropical $D$-Gröbner bases as for classical Gröbner bases.

Then, we provide with Algorithm \ref{trop row ech} a tropical row-echelon form computation algorithm for Macaulay matrices. We show that the F5 criterion still holds in a tropical setting. We therefore define the tropical Matrix-F5 algorithm (Algorithm \ref{trop MF5}) as an adaptation of a \textit{naïve} Matrix-F5 algorithm with the tropical row-echelon form computation. We then have the following result :

\begin{prop}
Let $(f_1,\dots, f_s ) \in K[X_1, \dots, X_n]^s$ be a sequence of homogeneous polynomials. Then, the tropical Matrix-F5 algorithm computes a tropical $D$-Gröbner basis of $\left\langle f_1,\dots,f_s \right\rangle$. Time-complexity is in $O \left( s^2 D \binom{n+D-1}{D}^3 \right)$ operations in $K$, as $D \rightarrow +\infty.$\footnote{One could also write $O \left( s^2 \binom{n+D}{D}^3 \right) $.} If $(f_1,\dots, f_s )$ is regular, time-complexity is in $O \left( s D \binom{n+D-1}{D}^3 \right).$
\end{prop}

The Macaulay bound is also available. Furthermore, not only does the tropical Matrix-F5 algorithm computes tropical $D$-Gröbner bases, but it is compatible with finite-precision coefficients, under the assumption that the entry sequence is regular.
Let us assume that $K=\mathbb{Q}_p,$ $\mathbb{F}_q \llbracket t \rrbracket$ or $\mathbb{Q}\llbracket t \rrbracket.$ Let $(f_1,\dots, f_s ) \in K[X_1, \dots, X_n]^s.$ We define a bound on the precision, $prec_{MF5trop} \left(( f_1, \dots, f_s ), D, \geq \right), $ and one on  the loss in precision, $loss_{MF5trop}\left((f_1, \dots, f_s ), D, \geq \right) ,$ which depend explicitly on the coefficients of the $f_i$'s. Then we have the following proposition regarding to numerical stability of tropical Gröbner bases :

\begin{prop} \label{Numerical Stability}
Let $F=(f_1,\dots, f_s ) \in K[X_1, \dots, X_n]^s$ be a  regular sequence of homogeneous polynomials. 

Let $(f_1',\dots,f_s')$ be some approximations of $F,$ with precision $l$ on their coefficients  better than $prec_{MF5trop}(F, D, \geq).$ 
Then, with the tropical Matrix-F5 algorithm, one can compute an approximation $g_1',\dots, g_t'$ of a Gröbner basis of $\left\langle F \right\rangle,$ up to precision $l-loss_{MF5trop}(F, D, \geq).$
\end{prop}

This contrasts with the case of classical Gröbner bases, for a monomial order $\omega$, over $p$-adics (or complete discrete valuation fields) considered in \cite{Vaccon}. Indeed, the structure hypothesis \textbf{H2} which requires that the ideals$ \left\langle  f_1, \dots ,f_i \right\rangle$ are weakly-$\omega$ is no longer necessary (see Subsection \ref{Comparison with classical GB}). It is only replaced by the (possibly stronger) assumption that the initial precision is better than $prec_{MF5trop}(F, D, \geq).$ In the special case of a weight $w=(0,\dots,0),$ the loss in precision is the smallest linear algebra on the Macaulay matrices can provide, and numerical evidences show that it is in average rather low.

Finally, we show that a faster variant of Matrix-F5 algorithm, where one use the Macaulay matrices in degree $d$ to build the Macaulay matrices in degree $d+1,$  can be adapted to compute tropical Gröbner bases. We first provide a tropical LUP-form computation for Macaulay matrices that is compatible with signatures, and then what we call the tropical signature-based Matrix-F5 algorithm (algorithms \ref{trop LUP} and \ref{sig based tropMF5}). We prove the following result :

\begin{prop}
Let $(f_1,\dots, f_s ) \in K[X_1, \dots, X_n]^s$ be a sequence of homogeneous polynomials. Then, the tropical signature-based Matrix-F5 algorithm computes a tropical $D$-Gröbner basis of $\left\langle f_1,\dots,f_s \right\rangle.$
\end{prop}

Time-complexity is then in $O \left( s D \binom{n+D-1}{D}^3 \right)$ operations in $K,$ as $D \rightarrow +\infty$ and $O \left(  D \binom{n+D-1}{D}^3 \right)$ when the input polynomials form a regular sequence.

\textbf{Structure of the paper:} Section 2 is devoted to provide a tropical setting and definitions for tropical Gröbner bases. In Section 3, we show that matrix algorithms can be performed to compute such bases. To that intent, after an introduction to matrix algorithms for Gröbner bases, we provide a row-reduction algorithm that will make a first \textit{naïve} Matrix-F5 algorithm available. We then prove and analyze this tropical Matrix-F5 algorithm. In Section 4 we analyze the stability of this algorithm over inexact fields with valuations, such as $\mathbb{Q}_p$. Section 5 is devoted to numerical examples regarding the loss in precision in the computation of tropical Gröbner bases. In Section 6, we prove that the classical signature-based Matrix-F5 algorithm is available, along with an adapted tropical LUP algorithm for row-reduction of Macaulay matrices. Finally, Section 7 is a glance at some future possible developments for tropical Gröbner bases.

\section{Context and motivations}


From now on, let $K$ be a field equipped with a valuation $val \: : \: K^* \rightarrow \mathbb{R}$. Let $R$ be the ring of integers of $K,$ $m$ its maximal ideal, $k_K$ its residue field and let $\Gamma = Im(val).$ An example of such a field is $\mathbb{Q}_p$ with $p$-adic valuation. In that case, $R = \mathbb{Z}_p$, $m = p \mathbb{Z}_p,$ $k_K=\mathbb{Z}/p\mathbb{Z}$ and $\Gamma = \mathbb{Z}.$

Let also $n \in \mathbb{Z}_{>0},$ $A=K[X_1, \dots, X_n],$ $B = R[X_1,\dots, X_n]$ and $C = k_K[X_1,\dots,X_n].$ We write $\vert f \vert$ for the degree of a homogeneous polynomial $f \in A,$ and $A_d =  K[X_1, \dots, X_n]_d$ for the $K$-vector space of degree $d$ homogeneous polynomials.

\subsection{Tropical varieties, tropical Gröbner bases}

If $I$ is an homogeneous ideal in $A$, and $V(I) \subset \mathbb{P}_K^{n-1}$ is the projective variety defined by $I$. Then the tropical variety defined by $I$, or the tropicalization of $V(I),$ is $Trop(I)=\overline{val \left( V(I) \cap (K^*)^{n} \right)}$ (closure in $\mathbb{R}^n$).
$Trop(I)$ is a polyhedral complex and acts as a combinatorial shadow of $V(I)$ : many properties of $V(I)$ can be recovered combinatorially from $Trop(I)$. 

If $w \in \Gamma^n$, we define an order on the terms of $K[X_1,\dots ,X_n].$

\begin{deftn}
If $a,b \in K$ and $x^\alpha,$ $x^\beta$ are two monomials in $A$, $a x^\alpha \geq_w b x^\beta$ if $val(a)+w \cdot \alpha \leq val(b) +w \cdot \beta$. Naturally, it is possible that $a x^\alpha \neq b x^\beta$ and $val(a)+w \cdot \alpha = val(b) +w \cdot \beta$.

For any $f \in A$, we can define $LT_{\geq_w}(f)$, and then $LT_{\geq_w}(I)$, for $I \subset A$ an ideal, accordingly. 
\end{deftn}

We remark that $LT_{\geq_w}(f)$ might be a polynomial (with more than one term).
For example, if we take $w = [1,2,3]$ in $\mathbb{Q}_2 [x,y,z]$ (with $2$-adic valuation), then \[LT_{\geq_w}\left( x^4+x^2y+2y^4+2^{-8}z^4 \right)=x^4+x^2y+2^{-8}z^4.\]

$Trop(I)$ is then connected to $LT_{\geq_w}(I)$:

\begin{theo}[Fundamental th. of tropical geometry] If $K$ is algebraically closed with non-trivial valuation, $Trop(I)$ is the closure in $\mathbb{R}^n$ of those $w \in \Gamma^n$ such that $LT_{\geq_w}(I)$ does not contain a monomial.
\end{theo}
\begin{proof}
See Theorem 3.2.5 of \cite{Maclagan}.
\end{proof}

To compute $LT_{\geq_w} (I)$ one can add a (classical) monomial order in order to break ties when $LT_{\geq_w}(f)$  has more than one monomial.

\begin{deftn} \label{trop order}
Let us take $\geq_1$  a monomial order on $A$.

Given $a,b \in K$ and $x^\alpha$ and $x^\beta$ two monomials in $A$, we write $a x^\alpha \geq b x^\beta$ if $val(a)+w \cdot \alpha < val(b) +w \cdot \beta$, or $val(a)+w \cdot \alpha = val(b) +w \cdot \beta$ and $x^\alpha \geq_1 x^\beta$.

Let $f \in A$ and $A$ be an ideal of $A.$ We define $LT(f)$ and $LT(I)$ accordingly. We remark that $LT(I)=LT_{\geq_1} (LT_w(I))$. We define $LM(f)$ to be the monomial of $LT(f),$ and $LM(I)$ accordingly.
If $G=(g_1,\dots g_s) \in A^s$ is such that its leading monomials $(LM(g_1),\dots,LM(g_s))$ generate $LM(I),$ we say that $G$ is a \textbf{tropical Gröbner basis} of $I.$
\end{deftn}

We can finally remark that to compute a generating set of $LM_{\geq_w}(I),$ it is enough to compute a tropical Gröbner basis of $I.$

\textbf{Comparison with notations in previous works:} In \cite{Chan-Maclagan}, $K$ is such that there is a group homomorphism $\phi : \Gamma \rightarrow K$ such that for any $w \in \Gamma,$ $val( \phi (w))=w.$ If $x \in R,$ its reduction modulo $m$ is denoted by $\overline{x}.$ We define $\rho: \: K^* \rightarrow k_K$ to be defined by $\rho (x) = \overline{x \phi ( val (x))}.$ $\rho$ extends naturally to $A \setminus \left\lbrace 0 \right\rbrace$ with $\rho (\sum_u a_u x^u) = \sum_u \rho (a_u) x^u.$ $\geq_1$ extends naturally to $C.$ Let $w \in \Gamma^n$, then, in \cite{Chan}, the author defines for any $f \in A,$ $in_w = \rho (LT_{\geq_w}(f))$ and $lm(f) = LM_{\geq_1} (in_w).$ Let $G=(g_1,\dots,g_s) \in A^s.$ Then $G$ is a tropical Gröbner basis of $I = \left\langle G \right\rangle$ for the term order $\leq$ if and only if $(in_w(g_1), \dots , in_w (g_s))$ is a Gröbner basis of $in_w(I)$ for $\leq_1.$ As a consequence, computing $LM(I)$ and $in(I)$ yields the same monomials. Nevertheless, we prefer working with $LM$ since computations over (inexact) fields with valuations are among our motivations.


\subsection{The algorithm of Chan and Maclagan}

\textbf{A Buchberger-style algorithm: } in their article \cite{Chan-Maclagan}, Chan and Maclagan have proved that if one modifies the classical division algorithm of a polynomial by a finite family of polynomials with a variant of Mora's tangent cone algorithm, then one can get a division algorithm suited to the computation of tropical Gröbner bases. Indeed, they proved that Buchberger's algorithm using this division algorithm computes tropical Gröbner bases of ideals generated by homogeneous polynomials. The main ideas of their division algorithm is to allow division by previous partial quotients, and choose the divisor polynomial with a suited \textit{écart} function.


\textbf{Precision issues: } Polynomial computation over (inexact) fields such as $\mathbb{Q}_p$ or $\mathbb{F}_p \llbracket t \rrbracket$ is our main motivation. To compute tropical Gröbner bases in such a setting, one may want to apply Chan and Maclagan's algorithm. Unfortunately, Buchberger-style algorithms rely on zero-testing: the termination criterion is Buchberger's. This is definitely not suited to finite precision. For instance, let $F$ be $(x^2+xy+y^2+(1+O(p^N))t^2,x^2+2xy+4y^2+(1+O(p^N))t^2,t^4) \in \mathbb{Q}_p[x,y,t]^3$, for some $N\in \mathbb{N}.$ Then the application of Chan and Maclagan's algorithm (\textit{e.g.} for $w=(0,0,0)$ and grevlex) lead to $S$-polynomials that reduce to quantity of the form $O(\pi^{N'})xyt^2$, \textit{i.e.} such that it is not possible to decide whether the polynomial in remainder is zero or not. Such issues appear even with the usage of Buchberger's criteria. Hence, they exclude the usage of Buchberger-style algorithms for most of the computations of Gröbner bases over fields such as $\mathbb{Q}_p$. 

\section{A tropical Matrix-F5 algorithm}

\subsection{Matrix algorithm} 

Here we show that to compute a tropical Gröbner basis of an ideal given by a finite sequence of homogeneous polynomials, a matrix algorithm can be written. The first main idea is due to Daniel Lazard in \cite{Lazard}, who remarked that for an homogeneous ideal $I \subset A$, generated by homogeneous polynomials $(f_1, \dots, f_s)$, for $d \in \mathbb{N}$, then as $K$-vector space: $I \cap A_d= \left\langle x^\alpha f_i, \: \vert \alpha \vert + \vert f_i \vert=d \right\rangle. $ One of the main features of this property is that it can be given in term of matrices. First, we define the matrices of Macaulay : 

\begin{deftn} \label{Macaulay}
Let $B_{n,d}$ be the basis of the monomials of degree $d$, ordered decreasingly according to $\geq.$ Then for $f_1, \dots, f_s \in A$ homogeneous polynomials, $\vert f_i \vert = d_i$, $d \in \mathbb{N}$, we define $Mac_d(f_1, \dots, f_s)$ to be the matrix with coefficients in $K$ and whose rows are  $x^{\alpha_{1,1}}f_1,\dots,x^{\alpha_{1,\binom{n+d-d_1-1}{n-1}}},$ $x^{\alpha_{2,1}} f_2,\dots,$ $x^{\alpha_{s,\binom{n+d-d_s-1}{n-1}}} f_s,$ written in the basis $B_{n,d}.$ The  $x^{\alpha_{i,1}}<\dots <x^{\alpha_{i,\binom{n+d-d_i-1}{n-1}}}$'s  are the monomials of degree $ n+d-d_i-1.$ The $i$-th column of this matrix corresponds to the $i$-th monomial of $B_{n,d}.$
\end{deftn}

If we identify naturally the rows vectors of $k^{\tbinom{n+d-1}{n-1}}$ with homogeneous polynomials of degree $d$, then \[Im(Mac_d(f_1, \dots, f_s))=I \cap A_d,\] with $Im$ being the left image of the matrix. 

When performing classical matrix algorithms to compute Gröbner bases (see \cite{Bardet}), the idea is then to compute row-echelon forms of the $Mac_d(f_1,\dots,f_s)$ up to some $D$: if $D$ is large enough, the reduced rows forms a Gröbner basis of $I$. Though, it is not easy to guess in advance up to which $D$ we have to perform row-reductions of Macaulay matrices. This is why the idea of tropical $D$-Gröbner bases can be introduced.

\begin{deftn}
Let $I$ be an ideal of $A$, Then $(g_1, \dots , g_l)$ is a $D$-Gröbner basis of $I$ for $\geq$ if for any $f \in I$, homogeneous of degree less than $D$, there exists $1\leq i \leq l$ such that $LT(g_i)$ divides $LT(f)$.
\end{deftn}

\subsection{Tropical row-echelon form computation}

This Subsection is devoted to provide an algorithm that can compute $LM( \left\langle f_1, \dots, f_i \right\rangle) \cap A_d$ by computing echelonized bases of the $Mac_d(f_1, \dots, f_i).$ To track what the leading term of a row is, we add a label of monomials to the matrices:

\begin{deftn}
We define a Macaulay matrix of degree $d$ in $A$ to be a couple $(M,mon)$ where $M \in K^{r \times \binom{n+d-1}{n-1}}$ is a matrix, and $mon$ is the list of the $\binom{n+d-1}{n-1}$ monomials of degree $d$ of $A$, in decreasing order according to $\geq.$
If $mon$ is not ordered, $(M,mon)$ is only called a labeled matrix.
\end{deftn}
Algorithm \ref{trop row ech} over Macaulay matrices computes by pivoting the leading terms of their rows:

\IncMargin{1em}
\begin{algorithm} \label{trop row ech}
 \DontPrintSemicolon

 \SetKwInOut{Input}{input}\SetKwInOut{Output}{output}

 \Input{$M$, a Macaulay matrix of degree $d$ in $A= K[X_1, \dots, X_n]$, with $n_{row}$ rows and $n_{col}$ columns.}
 \Output{$\widetilde{M}$, the tropical row-echelon form of $M$}

$\widetilde{M} \leftarrow M$ ;  \;
\eIf{$n_{col}=1$ or $n_{row}=0$ or $M$ has no non-zero entry}{
				Return $\widetilde{M}$  ;\;
	}{			
Find $i,j$ such that $\widetilde{M}_{i,j}$ has the greatest term $\widetilde{M}_{i,j} x^{mon_j}$ (with smallest $i$ in case of tie); \;
Swap the columns $1$ and $j$ of $\widetilde{M}$, and the $1$ and $j$ entries of $mon$; \;
Swap the rows $1$ and $i$ of $\widetilde{M}$; \;
By pivoting with the first row, eliminate the first-column coefficients of the other rows ; \;
Proceed recursively on the submatrix $\widetilde{M}_{i \geq 2, j \geq 2}$; \;
Return $\widetilde{M}$; \;}

 \caption{The tropical row-echelon algorithm}
\end{algorithm}
\DecMargin{1em}

\begin{deftn}
We define the tropical row-echelon form of a Macaulay matrix $M$ as the result of the previous algorithm, and denote it by $\widetilde{M}.$ $\widetilde{M}$ is indeed in row-echelon form.
\end{deftn}

\textbf{Correctness:} $\widetilde{Mac_d (f_1, \dots , f_i)}$ provides exactly the leading terms of $\left\langle f_1, \dots, f_i \right\rangle \cap A_d$:

\begin{prop} \label{Correctness row-echelon}
Let $F=(f_1, \dots, f_s)$ be homogeneous polynomials in $A.$ Let $d \in \mathbb{Z}_{> 0}$ and $M=Mac_d(f_1,\dots,f_s).$. Let $I= \left\lbrace F \right\rbrace$ be the ideal generated by the $f_i$'s.

Let $\widetilde{M}$ be the tropical row-echelon form of $M$. Then the rows of $\widetilde{M}$ form a basis of $I \cap A_d$ such that their $LT$'s corresponds to $LT(I) \cap A_d$.
\end{prop}
The fact that the rows of $\widetilde{M}$ form a basis of $I \cap A_d$ is clear, it forms an echelonized basis (considering the basis $mon$ of $A_d$).
Considering the initial terms of $I \cap A_d$, the result is a direct consequence of the following lemma:
\begin{lem}
if $a x^\alpha > b_1 x^\beta$ and $a x^\alpha > b_2 x^\beta$, then $a x^\alpha >(b_1+b_2) x^\beta$.
\end{lem}

\textbf{Consequence:} We can find all the polynomials of a tropical $D$-Gröbner basis of $\left\langle f_1, \dots, f_s \right\rangle$ by computing the tropical row-echelon forms of the $Mac_d (f_1, \dots, f_s)$ for $d$ from $1$ to $D$. Nevertheless, there is room for improvement: those matrices are huge and most of the time not of full rank.

\subsection{The F5 criterion}

We introduce here Faugère's F5 criterion that is enough to discard most of the rows of the $Mac_d (f_1, \dots, f_s)$'s that do not yield any meaningful information for the computation of $LT(I).$ For any $j \in \llbracket 1,s \rrbracket,$ we denote by $I_j$ the ideal $\left\langle f_1, \dots, f_j \right\rangle.$ Then, Faugère proved in \cite{F02} that for a classical monomial ordering, if we know which monomials $x^\alpha$ are in $LM(I_{i-1}),$ we are able to discard corresponding rows $x^\alpha f_i$ of the Macaulay matrices. This criterion is compatible with our definition of $LM$:

\begin{theo}[F5-criterion] \label{Generateness}
For any $i \in \llbracket 1,s \rrbracket,$ 
\begin{align*}
I_i \cap & A_d =Span ( \{ x^\alpha f_k, \text{ s.t. } 1 \leq k \leq i, \: \vert x^\alpha f_k \vert =d \\
 & \:  \text{ and } x^\alpha \notin LM(I_{k-1}) \} ). \\
\end{align*}

\end{theo}

To prove this result, one can rely on the following fact, which can be proved inductively. Let $(f_1, \dots, f_i)$ be homogeneous polynomials of $A$ of degree $d_1, \dots, d_i$. Let $a_{\alpha_1} x^{\alpha_1}, \dots,$ $ a_{\alpha_u} x^{\alpha_u}$ be the initial terms of the rows of \linebreak $\widetilde{Mac_{d-d_i}(f_1, \dots, f_{i-1})}$, ordered by decreasing order (regarding the initial term). Let $x^{\beta_j}$ denote the remaining monomials of degree $d-d_i$ (\textit{i.e.} the monomials that are not an initial monomial of $\left\langle f_1, \dots, f_{i-1} \right\rangle \cap A_{d-d_i}$). Then, for any $k,$ the row $x^{\alpha_k} f_i$ of $Mac_d (f_1, \dots, f_i)$ is a linear combination of some rows of the form $x^{\alpha_{k+k'}} f_i$ ($k'>0$), $x^{\beta_j} f_i$ and $x^\gamma f_j$ ($j<i$) of $Mac_d(f_1, \dots, f_i)$.

Thus, it is now clear which rows we can remove with the F5 criterion. The following subsection provides an effective way of taking advantage of this criterion.

\subsection{A first Matrix-F5 algorithm}

\textbf{The tropical MF5 algorithm:} We apply Faugère's idea (see \cite{Bardet},\cite{Bardet3}, \cite{F02}) to the tropical setting and therefore provide a tropical Matrix-F5 algorithm:

\IncMargin{1em}
\begin{algorithm} \label{trop MF5}
\DontPrintSemicolon

 \SetKwInOut{Input}{input}\SetKwInOut{Output}{output}

 \Input{$F=(f_1, \dots, f_s) \in A^s$, homogeneous with respective degrees $d_1 \leq  \dots \leq d_s$, and $D \in \mathbb{N}$}
 \Output{$(g_1, \dots, g_k) \in A^k$, a $D$-tropical Gröbner basis of $\left\lbrace F \right\rbrace.$ }

$G \leftarrow F$  \;
\For{$d \in \llbracket 0, D \rrbracket$}{
	$\widetilde{\mathscr{M}_{d,0}}:=\emptyset$ \;
	\For{$i \in \llbracket 1,s \rrbracket$}{
		$\mathscr{M}_{d,i}:=\widetilde{\mathscr{M}_{d,i-1}}$ \;
		\For{ $\alpha$ such that $\vert \alpha \vert +d_i=d$}{
			 \If{$x^\alpha$ is not the leading term of a row of $\widetilde{\mathscr{M}_{d-d_i,i-1}}$}{
				Add $x^\alpha f_i$ to $\mathscr{M}_{d,i}$ \;
			}
		}
		Compute $\widetilde{\mathscr{M}_{d,i}}$, the tropical row-echelon form of $\mathscr{M}_{d,i}$ \;
		Add to $G$ all the rows with a new leading monomial. \;
	}
}
Return $G$ \;
 \caption{A tropical Matrix-F5 algorithm}
\end{algorithm}
\DecMargin{1em}

\textbf{Correctness:} What we have to show is that for any $d \in \llbracket 0, D \rrbracket$ and $i\in \llbracket 1,s \rrbracket$, $Im(\mathscr{M}_{d,i})=I_i \cap A_d.$ This can be proved by induction on $d$ and $i$. We remark that there is nothing to prove for $i=1$ and any $d$. Now let us assume that there exists some $i \in \llbracket 1,s \rrbracket$ such that for any $j$ with $1 \leq j < i$ and for any $d$, $0 \leq d \leq D$, $Im(\mathscr{M}_{d,j})=I_j \cap A_d.$
Then, $i$ being given, the first $d$ such that $\mathscr{M}_{d,i} \neq \mathscr{M}_{d,i-1}$ is $d_i.$ Let $d$ be such that $d_i \leq d \leq D.$ Then, with the induction hypothesis and corollary \ref{Generateness} :
\begin{equation} \label{Mac d i}
I_i \cap A_d = Im(\mathscr{M}_{d,i-1}) \begin{huge} + \end{huge} Span \left( \left\lbrace x^\alpha f_i, \text{ s.t. }  x^\alpha \notin LM(I_{i-1}) \right\rbrace \right).
\end{equation}

Besides, by the induction hypothesis and the correctness of the row-echelon algorithm (see Proposition \ref{Correctness row-echelon}), the leading terms of $I_{i-1} \cap A_{d-d_i}$ are exactly the leading terms of rows of $\widetilde{\mathscr{M}_{d-d_i,i-1}}.$ Thus, the rows that we add to $\widetilde{\mathscr{M}_{d,i-1}}$ in order to build $\mathscr{M}_{d,i}$ are exactly the $x^\alpha f_i$,  such that $ x^\alpha \notin LM(I_{i-1}).$ Finally, we remark that $Im(\mathscr{M}_{d,i})=Im(\widetilde{\mathscr{M}_{d,i-1}}).$ Therefore, $Im(\mathscr{M}_{d,i})$ contains both summands of (\ref{Mac d i}), and since it is clearly included in $I_i \cap A_d$, we have proved that $I_i \cap A_d = Im(\mathscr{M}_{d,i}).$ To conclude the correctness of the tropical MF5 algorithm, we point out that the correctness of the tropical row-echelon computation (see prop \ref{Correctness row-echelon}) show that the leading terms of rows of $\widetilde{\mathscr{M}_{d,i}}$ indeed correspond to the leading terms of $I_i \cap A_d.$
 
\subsection{Regular sequences and complexity}

\textbf{Principal syzygies and regularity:} The behavior of this algorithm with respect to principal syzygies is the same as the classical Matrix-F5 algorithm. See \cite{Bardet} for a precise description of the link between syzygies and row-reduction. We instead only prove the main result connecting principal syzygies and tropical row-reduction of Macaulay matrices.

\begin{prop}
If a row reduces to zero during the tropical row-echelon form computation of the tropical MF5 algorithm, then the syzygy it yields is not in the module of principal syzygies.
\end{prop}
\begin{proof}
Let $\sum_{j=1}^i a_j f_j$ with $a_j \in A$ be a syzygy of $(f_1,\dots,f_i).$ If $a_j \neq 0$ and if this this syzygy is principal, then $a_i \in I_{i-1}$ and $LM(a_i) \in LM(I_{i-1}).$
Since because of the F5 criterion, there is no row of the form $x^\alpha f_i$ with $x^\alpha \in LM(I_{i-1})$ in the operated $\mathscr{M}_{d,i},$ then no such syzygy can be produced during the reduction of $\mathscr{M}_{d,i}.$ \end{proof}

\begin{cor} \label{injectivity}
If the sequence $(f_1, \dots, f_s)$ is regular, then no row of a Macaulay matrix in the tropical MF5 algorithm reduces to zero. In other words, the $\mathscr{M}_{d,i}$ are all injective, and have non-strictly less rows than columns.
\end{cor}
\begin{proof}
For a regular sequence of homogeneous polynomials, all syzygies are principal. See \cite{Elkadi} page 69.
\end{proof}

\textbf{Complexity:} The complexity to compute tropical row-echelon form of a matrix of rank $r$ with $n_{rows}$ rows and $n_{cols}$ columns can be expressed as $O(r \times n_{rows} \times n_{cols})$ operations in $K.$  
This yields the following complexities for Algorithm \ref{trop MF5}:
\begin{itemize}
\item $O \left( s^2D \binom{n+D-1}{D}^3 \right)$ operations in $K,$ as $D \rightarrow +\infty.$
\item $O \left( sD \binom{n+D-1}{D}^3 \right)$ operations in $K,$ as $D \rightarrow +\infty,$ in the special case where $(f_1, \dots f_s)$ is regular, because of corollary \ref{injectivity}.
\end{itemize}
Compared to the classical case, for which we refer to \cite{Bardet3}, complexity gets essentially an extra factor $s.$ This comes from the fact that we need to compute the tropical row-echelon form from start for each new $\mathscr{M}_{d,i}.$ In other words, we do not take into account the fact that, after building $\mathscr{M}_{d,i},$ $\widetilde{\mathscr{M}_{d,i-1}}$ was already under row-echelon form.

\textbf{Bound on $D$:} Regarding bounds on a sufficient $D$ for $D$-Gröbner bases to be Gröbner bases, we might not hope better bounds than in the classical case (\textit{i.e.} with trivial valuation) exist. Chan has proved in \cite{Chan} (Theorem 3.3.1) that $D = 2 (d^2/2+d)^{2^{n-2}},$ with $d=\max_i d_i,$ is enough.
If $(f_1,\dots,f_n)$ is a regular sequence, we remark that all monomials of degree greater than the Macaulay bound $\sum_i (d_i-1)+1$ are in $LM(I).$ This is a consequence of the fact that we know what is the Hilbert function of a regular sequence. Hence,

\begin{prop} \label{Macaulay bound}
If $(f_1,\dots,f_n) \in A^n$ is a regular sequence of homogeneous polynomials, all $D$-Gröbner bases are Gröbner bases for $D \geq \sum_i ( \vert f_i\vert-1)+1.$
\end{prop}

\section{The case of finite-precision CDVF}

\subsection{Setting}

Throughout this section, we further assume that $K$ is a complete discrete valuation field. We refer to Serre \cite{Serre} for an introduction to such fields. Let $\pi \in R$ be a uniformizer for $K$ and let $S_K \subset R$ be a system of representatives of $k_K=R/m.$ All numbers of $K$ can be written uniquely under its $\pi$-adic power series development form : $\sum_{k \geq l} a_k \pi^k$ for some $ l \in \mathbb{Z}$, $a_k \in S_K$. We assume that $K$ is not an exact field, but $k_K$ is, and symbolic computation can only be performed on truncation of $\pi$-adic power series development. We denote by finite-precision CDVF such a field. An example of such a CDVF is $K = \mathbb{Q}_p$, with $p$-adic valuation. 
We are interested in the computation of tropical Gröbner bases over finite-precision CDVF and its comparison with that of classical Gröbner bases.

\subsection{Precision issues with leading terms}
For any $m \in \mathbb{Z},$ let $O(\pi^m)=\pi^m R.$ In a finite-precision CDVF $K$, we are interested in computation over approximations $x$ of elements of $K$ which take the form \linebreak$x=\sum_{k \geq l}^{m-1} a_k \pi^k +O(\pi^m).$ $m$ is called the precision over $x.$

If the precision on the coefficients of $f \in A$ is not enough, then one can not determine what the leading term of $f$ is. For example, on $\mathbb{Q}_p[X_1,X_2]$, with $w=(0,4)$ and lexicographical order, then one can not compare $O(p^2)X_1$ and $X_2$. Yet, with enough precision, such an issue does not occur when computing tropical row-echelon form. The following proposition provides a bound on the precision needed on $f$ to determine its leading term.

\begin{prop} \label{Prec needed}
Let $f \in A$ be an homogeneous polynomial, and let $a X^\alpha$ be its leading term.

Then precision $val(a)+\max_{\vert \beta \vert =d} \left( (\alpha -\beta) \cdot w \right)$ on the coefficients of $f$ is enough to determine which term of $f$ is $LT(f).$
\end{prop}
\begin{proof}
We only have to remark that $O(p^n) X^\beta < a X^\alpha$ if and only if $n>val(a)+(\alpha - \beta) \cdot w$.
\end{proof}


\subsection{Row-echelon form computation}

\textbf{Regular sequences:} As we have already seen, when dealing with finite-precision coefficients, a crucial issue is that one can not decide whether a coefficient $O(\pi^k)$ is zero or not.
Fortunately, thanks to Corollary \ref{injectivity}, when the input polynomials form a regular sequence, all matrices in the tropical MF5 algorithm are injective. It means that if the precision is enough, the tropical row-echelon form computation performed over these matrices will have no issue with finding pivots and deciding what the leading terms of the rows are. In other words, if the precision is enough, \textbf{there is no zero-testing issue}.

We then estimate which precision is enough in order to be able to compute $D$-Gröbner bases of such a sequence.

\textbf{A sufficient precision:} 
\begin{prop} \label{prec tro row-echelon}
Let $M$ be an injective tropical Macaulay matrix with coefficients in $R$, of degree $d$. Let $a_1, \dots , a_u$ be the pivots chosen during the computation of its tropical row-echelon form. Let $x^{\alpha_k}$ be the corresponding monomials.
Let $prec$ be :
\[prec = \sum_{k} val(a_k)+\max_k val(a_k)+\max_{k, \vert \beta \vert =d} \left( \alpha_k - \beta \right) \cdot w.\]

Then, if the coefficients of the rows are known up to the same precision $O(\pi^{prec}),$ the tropical row-echelon form computation of $M$ can be computed, and the loss in precision is $\sum_{k} val(a_k).$
\end{prop}
\begin{proof}

We begin with a matrix $M$ with coefficients all known with precision $O(\pi^{l}),$ and we first assume that there is no issue with finding the pivots. Thus, we first analyze what the loss in precision is when we pivot. That is, we wish to put a \lq\lq{}real zero\rq\rq{} on the coefficient $M_{i,j}= \varepsilon \pi^{n_1}+O(\pi^n)$, by pivoting with a pivot $piv=\mu \pi^{n_0}+O(\pi^n)$ on row $L$, with $n_0, n_1 <n $ be integers, and $\varepsilon=\sum_{k=0}^{n-n_1-1} a_k\pi^k$, $\mu =\sum_{k=0}^{n-n_0-1} b_k \pi^k$, with $a_k,b_k \in S_K$, and $a_0,b_0 \neq 0.$ We remark that by definition of the pivot, necessarily, $n_0 \leq n_1.$ Now, this can be performed by the following operation on the $i$-th row $L_i$ :
\[L_i\leftarrow L_i-\frac{M_{i,j}}{piv}L=L_i+(\varepsilon \mu^{-1} \pi^{n_1-n_0}+O(\pi^{n-n_0}))L,\]
along with the symbolic operation $M_{i,j}\leftarrow 0.$ Indeed, $\frac{M_{i,j}}{piv}=\frac{\varepsilon \pi^{n_1}+O(\pi^{n})}{\mu \pi^{n}+O(\pi^{m_0})},$ therefore $\frac{M_{i,j}}{piv}=\varepsilon \mu^{-1} \pi^{n_1-n_0}+O(\pi^{n-n_0}).$ As a consequence, after the first pivot is chosen and other coefficient of the first column have been reduced to zero, the coefficients of the submatrix $\widetilde{M}_{i\geq 2, j \geq 2}$ are known up to \linebreak $O(\pi^{l-val(a_1)}).$
We can then proceed inductively to prove that after the termination of the tropical row-echelon form computation, coefficients of $\widetilde{M}$ are known up to \linebreak$O(\pi^{l-val(a_1 \times \dots \times a_u)}).$ Since we have to be able to determine what the leading terms of the rows are in order to determine what the pivots are, then, with Proposition \ref{Prec needed},  it is enough that  $l-val(a_1 \times \dots \times, a_u)$ is bigger than $\max_{k, \vert \beta \vert =d} \left( \alpha - \beta \right) \cdot w,$ which concludes the proof.
\end{proof}

\subsection{Tropical MF5 algorithm}

We apply this study of the row-echelon computation to prove Proposition \ref{Numerical Stability} concerning the tropical Matrix-F5 algorithm over CDVF. To facilitate this investigation, and only for section 4, the step $\mathscr{M}_{d,i} := \widetilde{\mathscr{M}_{d,i-1}}$ in algorithm \ref{trop MF5} is replaced with $\mathscr{M}_{d,i} := \mathscr{M}_{d,i-1}.$ This is harmless since both matrices have same dimension and image.
We first define bounds on the initial precision and loss in precision. Let $(f_1,\dots, f_s ) \in B^s$ be a  regular sequence of homogeneous polynomials.

\begin{deftn}
Let $d \geq 1$ and $1 \leq i \leq s.$ Let $x^{\alpha_1}, \dots, x^{\alpha_u}$ be the monomials of the leading terms of $\left\langle f_1 ,\dots, f_i \right\rangle \cap A_d.$

Let $\Delta_{d,i}$ be the minor over the columns corresponding to the $x^{\alpha_l}$ that achieves smallest valuation.
Let \[\Box_{d,i}=2\Delta_{d,i} + \max_{k, \vert \beta \vert =d} \left( \alpha_k - \beta \right) \cdot w.\]

We define $prec_{MF5trop}( ( f_1, \dots, f_s ), D, \geq ) = \max_{d \leq D,i} \Box_{d,i},$ and $loss_{MF5trop}(( f_1, \dots, f_s ), D, \geq) = \max_{d \leq D,i} \Delta_{d,i}.$
\end{deftn}

As a consequence of Proposition \ref{prec tro row-echelon}, these bounds are enough for Proposition \ref{Numerical Stability}.

Furthermore, we can precise the special case of $w=0$ :

\begin{prop} \label{weight zero}
If $w=0,$ then the loss in precision corresponds to the maximal minors of the $\mathscr{M}_{d,i}$ with the smallest valuation. In particular, $w=0$ corresponds to the smallest $loss_{MF5trop}$ and a straight-forward $prec_{MF5trop}.$
\end{prop}

\subsection{Precision versus time-complexity}

We might remark that if one want to achieve a smaller loss in precision, one might want to drop the F5 criterion and use the tropical row-reduction algorithm on the whole Macaulay matrices until enough linearly-free rows are found. The required number of rows can be computed thanks to the F5-criterion and corollary \ref{Generateness} if Macaulay matrices are operated iteratively in $d$ and $i.$ This way, one would be assured that its pivots will yield the smallest loss of precision possible over $Mac_d(f_1, \dots, f_s).$ Yet, such an algorithm would be more time-consuming because of the huge number of useless rows, and would be in $O \left( s^2D \binom{n+D-1}{D}^3 \right)$ operations in $K$ even for regular sequences.

\subsection{Comparison with classical Gröbner bases}
\label{Comparison with classical GB}

We compare here the results over finite-precision CDVF for computation of tropical Gröbner bases and for computation of classical Gröbner bases, as it was performed in \cite{Vaccon}.

We recall the main result of \cite{Vaccon} : 
\begin{deftn}
Let $\omega$ be a monomial order on $A$. Let $F=(f_1,\dots,f_s) \in B^s$ be homogeneous polynomials. Let $\mathscr{M}_{d,i}$ be the Macaulay matrix in degree $d$ for $(f_1, \dots, f_i),$ without the rows discarded by the F5-criterion. Let $l_{d,i}$ be the maximum of the $l \in \mathbb{Z}_{\geq 0}$ such that the $l$-first columns of $\mathscr{M}_{d,i}$ are linearly free. We define \[\Delta_{d,i}=\min\left( val \left( \left\lbrace \text{minor over the } l_{d,i} \text{-first columns of } \mathscr{M}_{d,i} \right\rbrace \right) \right).\]

We define the Matrix-F5 precision of $F$ regarding to $\omega$ and $D$ as :
\[prec_{MF5}(F, D, \omega) = \max_{d\leq D, \: 1\leq i \leq s} val \left( \Delta_{d,i} \right).\]
\end{deftn}

Then, $prec_{MF5}(F, D, \omega)$ is enough to compute approximate $D$-Gröbner bases :

\begin{theo}
Let $(f_1',\dots,f_s')$ be approximations of the homogeneous polynomials $F=(f_1,\dots,f_s) \in B^s,$ with precision better than $prec_{MF5}=prec_{MF5}(F, D, w).$
We assume that $(f_1,\dots,f_s)$ is a regular sequence (\textbf{H1}) and all the $\left\langle f_1, \dots, f_i \right\rangle$ are weakly-$\omega$-ideals (\textbf{H2}).
Then, the weak Matrix-F5 algorithm computes an approximate $D$-Gröbner basis of $(f_1',\dots,f_s'),$ with loss in precision upper-bounded by $prec_{MF5}.$
The complexity is in $O \left( sD \binom{n+D-1}{D}^3 \right)$ operations in $K,$ as $D \rightarrow +\infty.$
\end{theo}

We remark that for tropical Gröbner bases, the structure hypothesis \textbf{H2} is compensated by the precision requirement for the tropical row-echelon computation : $\max_k val(a_k)+\max_{k, \vert \beta \vert =d} \left( \alpha_k - \beta \right) \cdot w$ so that there is no position problem for the leading terms when a tropical Gröbner basis is computed. This leads to a bound on the required precision, $prec_{MF5trop}(F, D, \geq),$ that might be bigger than $prec_{MF5}$ but with no position problem and no requirement for \textbf{H2}. 

Thus, for tropical Gröbner bases over a CDVF (where the valuation is non-trivial), the only structure hypothesis is the regularity \textbf{H1}, and is clearly generic, whereas for classical Gröbner bases, \textbf{H1} and \textbf{H2} might be generic only in special cases, like for the grevlex ordering if Moreno-Socias' conjecture holds. Therefore, tropical Gröbner bases computation may require a bigger precision on the input than classical Gröbner bases, but it can be performed generically, while it is not clear for classical Gröbner bases.

Finally, when the weight $w$ is zero, thanks to Proposition \ref{weight zero}, the smallest loss in precision defined by minors of Macaulay matrices is attained. 

\vfill\eject

\section{Implementation}

A toy implementation in Sage \cite{Sage} of the previous algorithm is available at \url{http://perso.univ-rennes1.fr/tristan.vaccon/toy_F5.py}. The purpose of this implementation was the study of the precision. It is therefore not optimized regarding to time-complexity.
We have applied the tropical Matrix-F5 algorithm to homogeneous polynomials with varying degrees and random coefficients in $\mathbb{Z}_p$ (regarding to the Haar measure): $f_1,\dots, f_s,$ of degree $d_1, \dots, d_s$ in $\mathbb{Z}_p[X_1,\dots, X_s],$ known up to initial precision $30,$ with a given weight $w$ and the grevlex ordering to break the ties, and up to $D$ the Macaulay bound. We have done this experiment 20 times for each setting and noted maximal loss, mean loss in precision and the number of failures (\textit{i.e.} the computation can not be completed due to precision).  We have compared with the weak-MF5 of \cite{Vaccon} with grevlex on the same setting (the "grevlex" cases in the array). We present the results in the following array :
\begin{tabular}{|c|c|c|c|c|c|c|}
\hline
 $d =$ & $w$ & D & $p$ &   \tiny{maximal loss} & \tiny{mean loss} & \tiny{failure} \\  \hline
[3,4,7]& grevlex & 12 & 2 & 9   & 0.1 & 0\\ \hline
[3,4,7]& [1,-3,2] & 12 & 2 & 11   &0.1 & 0\\ \hline
[3,4,7]& [0,0,0] & 12 & 2 &  0  & 0 &0\\ \hline 
[3,4,7]& [1,-3,2] & 12 & 7 &  3 & .02 &0\\ \hline
[3,4,7]& [0,0,0] & 12 & 7 &  0  & 0 &0\\ \hline
[2,3,4,5] & grevlex &11 & 2 &  9 & 1.6 &2\\ \hline 
[2,3,4,5] & [1,4,1,-1] &11 & 2 &  13 & 0.2 &0\\ \hline 
[2,3,4,5] &  [0,0,0,0] & 11 & 2 &  0 & 0 &0\\ \hline 
[2,3,4,5] & [1,4,1,1] & 11 & 7 &  5 & 0.02 &0\\ \hline 
\end{tabular} 

These results suggest that the loss in precision is less when working with bigger primes. It seems reasonable since the loss in precision comes from pivots with positive valuation, whereas the probability that $val(x) = 0$ for $x \in \mathbb{Z}_p$ is $\frac{p-1}{p}.$ 
Those results also corroborate the facts that $w=[0,\dots ,0]$ lead to significantly smaller loss in precision.

\section{A faster tropical MF5 algorithm}

In this section, we show that one can perform in a tropical setting an adaptation of the classical, signature-based, Matrix-F5 algorithm presented in \cite{Bardet3}. This variant of the Matrix-F5 algorithm is characterized by the usage of the fact that $\widetilde{\mathscr{M}_{d,i}}$ is under echelon form to build a $\mathscr{M}_{d,i}$ closer to its echelon-form.

To that intent, we introduce labels and signatures for polynomials, and a tropical LUP-form computation.

\subsection{Label and signature}

\begin{deftn} \label{lab}
Given $(f_1, \dots, f_s) \in A^s$, a \textit{labeled polynomial} is a couple $(u,p)$ with $u=(l_1, \dots, l_s) \in A^s,$ $p \in A$ and $\sum_{i=1}^s l_i f_i=p$.

$u$ is called the \textit{label} of the labeled polynomial. We write $(e_1, \dots, e_s)$ to be the canonical basis of $A^s$.

If $u=(l_1, \dots, l_i, 0, \dots, 0)$ with $l_i \neq 0$, then the \textit{signature} of the labeled polynomial $(u,p),$ denoted by $sign((u,p)),$ is $(HM(l_i),i)$, with the following definition : $HM(l_i)$ is the highest monomial, regarding to $\leq$, that appears in $l_i$ with a non-zero coefficient.
\end{deftn}
\begin{rmk}
We must point out that in the definition of the signature, we \textit{do not} take into account the valuations of the coefficients in the label, hence the $HM(l_i)$ instead of $LT(l_i)$ or $LM(l_i)$. $HM(l_i)$ is not, in general, the monomial of the leading term of $l_i$.
\end{rmk}

\begin{deftn} \label{sign}
We define a total order on the set of signatures $\left\lbrace \text{monomials in } R \right\rbrace \times \left\lbrace 1 , \dots, s \right\rbrace$ with the following definition : $(x^\alpha,i) \leq (x^\beta,k)$ if $i<k$, or $x^\alpha \leq x^\beta$ and $i=k.$
\end{deftn}

Signatures are compatible with operations over labeled polynomials :

\begin{prop}
Let $(u,p)$ be a labeled polynomial, $(x^\alpha,i)=sign((u,l))$ and let $x^\beta$ be a monomial in $A$. Then \[sign((x^\beta u, x^\beta p))=(x^\alpha x^\beta,i).\]

If $(v,q)$ is another labeled polynomial such that \linebreak$sign((v,q)) < sign((u,p)),$ and if $\mu \in K,$ then $sign((u + \mu v, p+ \mu q))= sign((u,p)).$
\end{prop}

\subsection{Signature-preserving LUP-form computation}

From now on throughout this subsection, an additional datum will be attached to the rows of the Macaulay matrices: its label and signature. We make the further assumption that the rows are ordered with increasing signature. Such a matrix will be called a labeled Macaulay matrix. When adding a row, both its label and its signature will be noted, and all the operations on the rows are carried on to the labels of these rows.

\textbf{The algorithm:} We provide a tropical LUP algorithm for labeled Macaulay matrices to compute the leading term of the Macaulay matrices while preserving signatures.

\IncMargin{1em}
\begin{algorithm} \label{trop LUP}
 \DontPrintSemicolon

 \SetKwInOut{Input}{input}\SetKwInOut{Output}{output}

 \Input{$M$, a labeled Macaulay matrix of degree $d$ in $A$, with $n_{row}$ rows and $n_{col}$ columns.}
 \Output{$\widetilde{M}$, the $U$ of the tropical LUP-form of $M$}

$\widetilde{M} \leftarrow M$ ;  \;
\eIf{$n_{col}=1$ or $n_{row}=0$ or $M$ has no non-zero entry}{
				Return $\widetilde{M}$  ;\;
	}{			
\For{$i=1$ to $n_{row}$}{
Find $j$ such that $\widetilde{M}_{i,j}$ has the greatest term $\widetilde{M}_{i,j} x^{mon_j}$ over the row; \;
Swap the columns $1$ and $j$ of $\widetilde{M}$, and the $1$ and $j$ entries of $mon$; \;
By pivoting with the first row, eliminates the coefficients of the other rows on the first column; \;
Proceed recursively on the submatrix $\widetilde{M}_{i \geq 2, j \geq 2}$; \;}
Return $\widetilde{M}$; \;}

 \caption{The tropical LUP algorithm}
\end{algorithm}
\DecMargin{1em}

We remark that at the end of the algorithm, there exists a unipotent lower-triangular matrix $L$, a permutation matrix $P$, such that $\widetilde{M}=L M P,$ $\widetilde{M}$ is under row-echelon form up to permutation, and since we only add to a row $L_i$ a linear combination of rows that are above $L_i,$ those rows have a strictly lower signature than $L_i$, and therefore the operations performed on the rows (and on the columns) preserve the signature. Furthermore,
\begin{prop}
For any $1 \leq i \leq n_{row}(M),$ if $j$ is the index of the $i$-th row of $\widetilde{M}$, then $\widetilde{M}_{i,j} x^{\text{mon}_j}$ is the leading term of the polynomial corresponding to this row.
\end{prop}

Those remarks justify the name of tropical LUP algorithm, and the facts that this algorithm computes the leading terms of $Span(rows(M))$. Finally, since signature remains unchanged throughout the tropical LUP reduction, we can omit the labels and only handle Macaulay matrices on which the signatures of the rows are marked.

\subsection{A signature-based tropical MF5 algorithm}
We show that with LUP-reduction we can adapt the classical Matrix-F5 algorithm.

\textbf{The signature-based F5 criterion} is still available:
\begin{prop} \label{F5crit2}
Let $(u,f)$ be a labeled homogeneous polynomial of degree $d$, such that $sign(u)=x^\alpha e_i$, with $1 < i \leq s$ and $x^\alpha \in I_{i-1}$.
Then, \[x^\alpha \in Span \left( \left\lbrace x^\beta f_k, \vert x^\beta f_k \vert =d, \text{ and } (x^\beta,k)<(x^\alpha,i)\right\rbrace \right).\]

As a consequence, if $(u,f)$ is a labeled homogeneous polynomial of degree $d$ with $sign(u)=x^\alpha e_i$ and $x^\alpha \notin LM(I_{i-1})$. Then $f$ can be written $f=x^\alpha f_i + g$, with \[g \in Span \left( \left\lbrace x^\beta f_k, \vert x^\beta f_k \vert =d, \text{ and } (x^\beta,k)<(x^\alpha,i)\right\rbrace \right).\]
\end{prop}

\textbf{A faster tropical Matrix-F5 algorithm:}

\IncMargin{1em}
\begin{algorithm} \label{sig based tropMF5}
\DontPrintSemicolon

 \SetKwInOut{Input}{input}\SetKwInOut{Output}{output}

 \Input{$F=(f_1, \dots, f_s) \in A^s$, with respective degrees $d_1, \dots, d_s$, and $D \in \mathbb{N}$}
 \Output{$(g_1, \dots, g_k) \in A^k$, a $D$-tropical Gröbner basis of $\left\langle F \right\rangle$, if $D$ is large enough. }

$G \leftarrow F$  \;
\For{$d \in \llbracket 0, D \rrbracket$}{
	$\widetilde{\mathscr{M}_{d,0}}:=\emptyset$ \;
	\For{$i \in \llbracket 1,s \rrbracket$}{
		$\mathscr{M}_{d,i}:=\widetilde{\mathscr{M}_{d,i-1}}$ \;
		\For{$L$ a row of $\widetilde{\mathscr{M}_{d-1,i}}$}{
			 \For{$x \in \left\lbrace X_1, \dots ,X_n \right\rbrace$}{
			      $x^\alpha e_k := sign(xL)$ \;
			      \If{$k=i$, $x^\alpha$ is not the leading term of a row of $\widetilde{\mathscr{M}_{d-d_i,i-1}},$ and $\mathscr{M}_{d,i}$ has not already a row with signature $x^\alpha e_i$}{
				    Add $xL$ to $\mathscr{M}_{d,i}$. \;
			      }
			  }
		}
		Compute $\widetilde{\mathscr{M}_{d,i}}$, the tropical LUP-form of $\mathscr{M}_{d,i}$. \;
		Add to $G$ all the rows with a new leading monomial. \;
	}
}
Return $G$ \;

 \caption{The tropical signature-based Matrix-F5 algorithm}
\end{algorithm}
\DecMargin{1em}

\textbf{Correctness:} This algorithm indeed computes a tropical $D$-Gröbner basis. The first thing to prove is that with the building of the Macaulay matrices suggested in the algorithm, the two following properties are satisfied : $Im(\mathscr{M}_{d,i})=I_i \cap A_d$ and for any monomial $x^\alpha$ of degree $d-d_i$ such that $x^\alpha \notin LM(I_{i-1})$, $\mathscr{M}_{d,i}$ has a row with signature $x^\alpha e_i$. This can be proved by induction on $d$ and $i$.

Now, since the tropical LUP reduction indeed computes an echelon-basis of the $\mathscr{M}_{d,i}$, as in the previous tropical MF5 algorithm, the signature-based tropical MF5 algorithm computes tropical $D$-Gröbner bases.

\textbf{Complexity:} The main difference in complexity between Algorithm \ref{trop MF5} and Algorithm \ref{sig based tropMF5} is that for the latter, the computation of the tropical LUP-form of the $\mathscr{M}_{d,i+1}$ takes into account the fact that it was previously done on $\mathscr{M}_{d,i}$, \textit{i.e.} the first rows of $\mathscr{M}_{d,i+1}$ are already under row-echelon form with the right leading terms. As a consequence, the complexity to compute a tropical $D$-Gröbner basis of $(f_1,\dots, f_s)$  is the same as in the classical case, that is to say, $O \left( s D \binom{n+D-1}{D}^3 \right)$ operations in $K,$ as $D \rightarrow +\infty.$ If $(f_1,\dots, f_s)$ is a regular sequence, then the complexity is in $O \left( D \binom{n+D-1}{D}^3 \right).$

\section{Future works}
Since both Buchberger and Matrix-F5 algorithms are available, we conjecture that the F5 algorithm can be adapted to the tropical setting. It would probably reduce to adapt properly the TopReduction of \cite{F02}.

The numerical stability of Proposition \ref{Numerical Stability} and the fact that tropical Gröbner bases provide normal forms, suggest investigating the FGLM (\cite{FGLM}) algorithm to pass from a tropical order (with $w=(0,\dots,0)$) to a classical one, with a view toward stable computations over finite-precision CDVF.

\bibliographystyle{plain}

\end{document}